\documentclass[conference,10 pt,onecolumn]{IEEEtran}       
\IEEEoverridecommandlockouts                              

\usepackage[margin=0.9in]{geometry}

\usepackage{amsmath, amssymb, amsfonts}
\usepackage{amsthm}
\usepackage{mathtools}
\usepackage{cite}
\usepackage{graphicx}
\usepackage{hyperref, url}
\usepackage{tikz,pgfplots,pgfplotstable,xcolor}
\usetikzlibrary{matrix}
\usetikzlibrary{cd}
\usetikzlibrary{calc}
\usetikzlibrary{arrows}      
\usetikzlibrary{decorations.markings}
\usetikzlibrary{positioning}
\usepackage{booktabs}       
\usepackage{threeparttable}
\usepackage{subcaption} 
\usepackage{xcolor}	
\usepackage[nice]{nicefrac}

\makeatletter
\newcommand\footnoteref[1]{\protected@xdef\@thefnmark{\ref{#1}}\@footnotemark}
\makeatother

\newcommand{\Rb}{\mathbb{R}}
\newcommand{\Kcal}{\mathcal{K}}

\newcommand{\Acal}{\mathcal{A}}

\newcommand{\Scal}{\mathcal{S}}

\newcommand{\Ycal}{\mathcal{Y}}
\newcommand{\Zcal}{\mathcal{Z}}

\newcommand{\TCI}{\widetilde{CI}}
\newcommand{\TSI}{\widetilde{SI}}
\newcommand{\TUI}{\widetilde{UI}}
\newcommand{\norm}[1]{\left\lVert#1\right\rVert}
\newcommand\logeq{\mathrel{\vcentcolon\Longleftrightarrow}}
\newcommand{\uge}{\sqsupseteq}  
\newcommand{\mge}{\sqsupseteq'} 

\newcommand{\lnn}[1]{\ln\left(#1\right)}
\newcommand{\SK}[3]{S_{\rightarrow}\!\left({#1};{#2}\!\left|{#3} \right. \right)} 
\newcommand{\SKK}[3]{S_{\leftrightarrow}\!\left({#1};{#2}\!\left|{#3} \right. \right)} 
\newcommand{\Prv}[1]{\,{\rm Pr}\!\left[#1\right]}

\DeclareMathOperator*{\argmin}{arg\,min}

\DeclareMathOperator*{\conv}{conv}

\DeclareMathOperator*{\simu}{sim}
\DeclareMathOperator*{\bigtimes}{\textnormal{\Large $\times$}}

\theoremstyle{definition}
\newtheorem{theorem}{Theorem}
\newtheorem{lemma}[theorem]{Lemma}
\newtheorem{proposition}[theorem]{Proposition}
\newtheorem{corollary}[theorem]{Corollary}
\newtheorem{definition}[theorem]{Definition}
\newtheorem{example}[theorem]{Example}
\newtheorem*{example*}{Example~\ref{prop:vanishingUIequiv}b)}

\newtheorem{remark}[theorem]{Remark}

\definecolor{mahogany}{rgb}{0.65, 0., 0.5}
\definecolor{darkred}{rgb}{0.5,0,0}

\title{\LARGE \bf
Unique Informations and Deficiencies
}

\author{Pradeep Kr. Banerjee$^{\ast}$, Eckehard Olbrich, J\"urgen Jost, and Johannes Rauh
\thanks{$^{\ast}$The authors are with the Max Planck Institute for Mathematics in the Sciences, Leipzig, Germany.}
\thanks{{Email: \tt\small\{pradeep,olbrich,jjost,jrauh\}@mis.mpg.de}}
}

\begin{document}

\maketitle
\thispagestyle{plain} 
\pagestyle{plain}

\begin{abstract}

Given two channels that convey information about the same random variable, we introduce two measures of the unique information of one channel with respect to the other. The two quantities are based on the notion of generalized weighted Le Cam deficiencies and differ on whether one channel can approximate the other by a randomization at either its input or output. We relate the proposed quantities to an existing measure of unique information which we call the minimum-synergy unique information. We give an operational interpretation of the latter in terms of an upper bound on the one-way secret key rate and discuss the role of the unique informations in the context of nonnegative mutual information decompositions into unique, redundant and synergistic components.

\end{abstract}


\begin{center}
\fbox{\parbox{0.8\textwidth}{
  \textbf{Note:}
  The material in this manuscript has been presented at the Allerton conference 2018~\cite{Allerton2018}.  This manuscript contains some corrections: most notably, Lemma~18 was removed and Proposition~\ref{prop:vanishingUIequiv} was corrected.  The numbering of equations and results in this file agrees with the numbering of the published version.
}}
\end{center}

\section{Introduction}
Consider three random variables~$S$, $Y$, $Z$ with finite alphabets. Suppose that we want to know the value of~$S$, but we can only observe~$Y$ and~$Z$. The mutual information between~$S$ and~$Y$ can be decomposed into information that~$Y$ has about~$S$ that is \emph{unknown} to~$Z$ (we call this the \emph{unique} or \emph{exclusive} information of~$Y$ w.r.t.~$Z$) and information that~$Y$ has about~$S$ that is \emph{known} to~$Z$ (we call this the \emph{shared} or \emph{redundant} information).
\begin{align}\label{eq:MIdec1}
I(S; Y) = \underbrace{\TUI(S;Y \backslash Z)}_{\text{unique $Y$ wrt $Z$}}+\underbrace{\TSI(S; Y,Z)}_{\text{shared (redundant)}}.
\end{align}
Conditioning on~$Z$ annihilates the shared information but creates \emph{complementary} or \emph{synergistic} information from the interaction of~$Y$ and~$Z$. 
\begin{align}\label{eq:MIdec2}
I(S; Y|Z) = \underbrace{\TUI(S;Y \backslash Z)}_{\text{unique $Y$ wrt $Z$}}+\underbrace{\TCI(S; Y,Z)}_{\text{complementary (synergistic)}}.
\end{align}
Using the chain rule, the total information that the pair~$(Y,Z)$ conveys about~$S$ can be decomposed into four terms.
\begin{align}\label{eq:MIdec3}
I(S; YZ) &= I(S;Y)+I(S;Z|Y)\notag\\
&=\TUI(S;Y \backslash Z)+\TSI(S; Y,Z)+\TUI(S;Z \backslash Y)+\TCI(S; Y,Z),
\end{align}
where~$\TUI$, $\TSI$, and~$\TCI$ are nonnegative functions that depend continuously on the joint distribution of $(S,Y,Z)$. 
Nonnegative information decompositions of this form have been considered in~\cite{e16042161,HarderSalgePolani2013:Bivariate_redundancy,GriffithKoch2014:Quantifying_Synergistic_MI,WilliamsBeer}.

Any definition of the function~$\TUI$ fixes two of the terms in~\eqref{eq:MIdec3} which in turn also determines the other terms by (\ref{eq:MIdec1}) and (\ref{eq:MIdec2}). 
This gives rise to the \emph{consistency condition}: 
\begin{align} \label{eq:consistency}
	I(S;Y)+\TUI(S;Z \backslash Y)=I(S;Z)+\TUI(S;Y \backslash Z).
\end{align}
One can thus interpret the unique information as either the conditional mutual information without the synergy, or as the mutual information without the redundancy. The difference of the redundant and synergistic information is called the \emph{coinformation}~$CoI(S;Y;Z)$ which is symmetric in its arguments and can be negative: $CoI(S;Y;Z)=\TSI(S; Y,Z)-\TCI(S; Y,Z)=I(S; Y)-I(S; Y|Z)$~\cite{Bell2003}. Coinformation is a widely used measure in the neurosciences~\cite{syncode,LathamNirenberg05:Synergy_and_redundancy_revisited} with negative values being interpreted as synergy~\cite{kontoyiannis2005ISIT} and positive values as redundancy~\cite{syncode}. If the interactions induce redundancy and synergy in equal measure, then the coinformation cannot detect it. \emph{Correlational importance}, a nonnegative measure introduced in~\cite{LathamNirenberg05:Synergy_and_redundancy_revisited} (see also~\cite{minsyn2017}) to quantify the importance of correlations in neural coding is similar in spirit to the synergistic information. However, examples are known~\cite{schneidman2003synergy} when it can exceed the total mutual information.

The notions of synergy, redundancy and unique information also appear implicitly in information-theoretic cryptography~\cite{ahlswede1993,maurerintrinsic,renner2002ISIT,rauh2017secret}. Consider the source model for secret key agreement between distant Alice and Bob against an adversary, Eve~\cite{maurer1993,maurerintrinsic}. Alice, Bob and Eve observe i.i.d. copies of random variables~$S$, $Y$ and~$Z$ respectively, where~$(S,Y,Z)\sim P_{SYZ}$. Alice and Bob want to compute a secret key by communicating messages over a noiseless but insecure (public) channel transparent to Eve such that Eve's total information ($Z$ and the entire public communication) about the key is negligibly small. The maximum (asymptotic) rate at which Alice and Bob can compute a key is called the \emph{two-way secret key rate}~$\SKK{S}{Y}{Z}$. If Alice is allowed to use the public channel only once and Bob does not transmit, then the corresponding quantity is called the \emph{one-way secret key rate}~$\SK{S}{Y}{Z}$. 

An instance of a purely synergistic interaction is the \textsc{XOR} distribution: $Y$ and~$Z$ are independent binary random variables and~$S=Y+Z \mod 2$. Here~$I(S; YZ)=\TCI(S; Y,Z)=1$~bit. Clearly, if Alice can only see $S$ and Bob $Y$, then they cannot realize a secret key. However if Alice can also see~$Z$, then she can compute~$Y$ which can be used as a key perfectly secret from Eve since Eve's variable~$Z$ is independent of the key~$Y$. 

An instance of a purely redundant interaction is the \textsc{RDN} distribution: $S$, $Y$, $Z$ are uniformly distributed binary random variables with~$S=Y=Z$. Here~$I(S; YZ)=\TSI(S; Y,Z)=1$~bit. Alice and Bob cannot share a secret since Eve knows the exact values of~$S$ and~$Y$. 

Intuitively, if Bob has some unique information about Alice's variable~$S$ (that is not known to Eve), then there must be a situation where Bob can \emph{exploit} this information to his advantage. A distribution combining the \textsc{XOR} and \textsc{RDN} exemplifies such an advantage.
\begin{example} [\cite{RennerW03,maurerintrinsic}] \label{ex:sktoy}
	Consider the joint distribution
	\begin{center}
		\resizebox{.32\textwidth}{!}{%
		\centering
		\begin{tabular}{clcccc}
			\toprule
			\multicolumn{1}{c}{} & \multicolumn{4}{c}{$S$}\\
			\cmidrule{2-5}
			$Y$ ($Z$) & 0 & 1 & 2 & 3\\
			\midrule
			0 &  $\nicefrac{1}{8}$\ (0) & $\nicefrac{1}{8}$\ (1) &. &.\\
			1 &  $\nicefrac{1}{8}$\ (1) & $\nicefrac{1}{8}$\ (0) &. &.\\
			2 &.     &. &$\nicefrac{1}{4}$\ (2) &.\\
			3 &.     &. &. & $\nicefrac{1}{4}$\ (3)\\
			\bottomrule
		\end{tabular}}
	\end{center}	
	where $Z$'s value is shown in parentheses. For instance, the first entry of the table is read as~$P_{SYZ}(0,0,0)=\tfrac{1}{8}$. Here~$I(S;YZ)=2$, $\TSI(S; Y,Z)=1.5$, $\TCI(S; Y,Z)=0.5$ and $\TUI(S;Y\backslash Z)=0$\footnote{We compute the decomposition using a definition of the function~$\TUI$ proposed in~\cite{e16042161}. An efficient algorithm was recently proposed in~\cite{CUIfullver}.}.
	If Eve sees 2 or 3, she knows the exact values of~$S$ and~$Y$. When she sees 0 or 1, she can infer that Alice and Bob's values are in~$\{0,1\}$, but in this range, their observations are independent. Hence, no secret key agreement is possible. 
	
	Consider now the modified distribution
		\begin{center}
			\resizebox{.32\textwidth}{!}{%
				\centering
				\begin{tabular}{clcccc}
					\toprule
					\multicolumn{1}{c}{} & \multicolumn{4}{c}{$S$}\\
					\cmidrule{2-5}
					$Y$ ($Z$) & 0 & 1 & 2 & 3\\
					\midrule
					0 &  $\nicefrac{1}{8}$\ (0) & $\nicefrac{1}{8}$\ (1) &. &.\\
					1 &  $\nicefrac{1}{8}$\ (1) & $\nicefrac{1}{8}$\ (0) &. &.\\
					2 &.     &. &$\nicefrac{1}{4}$\ (0) &.\\
					3 &.     &. &. & $\nicefrac{1}{4}$\ (1)\\
					\bottomrule
			\end{tabular}}
		\end{center}	 
	where Eve's variable~$Z$ can only assume binary values. For this distribution, $I(S; YZ)=2$,~$\TSI(S; Y,Z)=0.5$, $\TCI(S; Y,Z)=0.5$ and~$\TUI(S; Y\backslash Z)=1$. Now Bob has \emph{unique information} about Alice's values w.r.t. Eve (namely, the ability to distinguish whether Alice sees values in the~\textsc{XOR} or the~\textsc{RDN} quadrant) which he can use to agree on 1~bit of secret. The \emph{intrinsic information}\footnote{The intrinsic information violates the consistency condition~\eqref{eq:consistency} and cannot be interpreted as unique information in our sense.} $I(S;Y\!\!\downarrow \!Z):=\min_{P_{Z'|Z}}I(S;Y|Z')$, a well-known upper bound on the two-way secret key rate~\cite{maurerintrinsic} is not tight in this toy example. It evaluates to 1.5~bits.
\end{example}

How can we decide if~$Y$ has some unique information about~$S$ (that is not known to~$Z$)? Consider the channels~$\kappa$ and~$\mu$ with the common input alphabet~$\Scal$ in Fig.~\ref{fig:1}. If~$\mu$ reduces to~$\kappa$ by adding a post-channel~$\lambda$ at its \emph{output}, then~$\mu$ may be said to \emph{include}~$\kappa$. One can draw the same conclusion for the channels~$\bar{\kappa}$ and~$\bar{\mu}$ with the common output alphabet~$\Scal$ in Fig.~\ref{fig:2}, if~$\bar{\mu}$ reduces to~$\bar{\kappa}$ by adding a pre-channel~$\bar{\lambda}$ at its \emph{input}. These are special cases of a \emph{channel inclusion} preorder first studied by Shannon~\cite{shannonorder}. In both these situations, one would expect that $Y$ provides no unique information about~$S$ w.r.t.~$Z$. A nonvanishing unique information would then quantify how far is one channel from being an inclusion or randomization of the other. 
\begin{figure}
	\centering
	\begin{subfigure}[b]{0.33\textwidth}
		\centering
		\begin{tikzpicture}[baseline=(current bounding box.center)]
		\node	(P) at	(-1,-0.5)  {$\pi_S\sim S$}; 
		\node	(S) at	(0,0)  {$\Scal$};
		\node	(Y)	 at	(3,0) {$\Ycal$};  
		\node	(S1) at	(0,-1) {$\Scal$}; 
		\node	(Y1) at	(3,-1) {$\Ycal$};  
		\node	(Z)	 at	(1.5,-1) {$\Zcal$}; 
		\draw[->,right, above] (S) to node {$\kappa$} (Y);
		\draw[->,right, above] (S1) to node {$\mu$} (Z);
		\draw[->,right, dashed, above] (Z) to node [xshift=0pt,yshift=0pt] {$\lambda$} (Y1);
		\end{tikzpicture}
		\caption{}
		\label{fig:1} 
	\end{subfigure}\qquad\qquad\qquad\qquad
   \begin{subfigure}[b]{0.33\textwidth}
		\centering
		\begin{tikzpicture}[baseline=(current bounding box.center)]
		\node	(P) at	(-1,-0.5)  {$\pi_Y\sim Y$}; 
		\node	(S) at	(3,0)  {$\Scal$};
		\node	(Y)	 at	(0,0) {$\Ycal$};  
		\node	(S1) at	(3,-1) {$\Scal$}; 
		\node	(Y1) at	(0,-1) {$\Ycal$};  
		\node	(Z)	 at	(1.5,-1) {$\Zcal$}; 
		\draw[->,right, above] (Y) to node {$\bar{\kappa}$} (S);
		\draw[->,right, dashed, above] (Y1) to node {$\bar\lambda$} (Z);
		\draw[->,right, above] (Z) to node [xshift=0pt,yshift=0pt] {$\bar{\mu}$} (S1);
		\end{tikzpicture}
		\caption{}
		\label{fig:2} 
	\end{subfigure}
	\caption{(a) Simulation of the channel~$\kappa$ by a randomization at the output of~$\mu$. $\kappa$~and~$\mu$ share a \emph{common input} alphabet~$\Scal$. (b) Simulation of the channel~$\bar\kappa$ by a randomization at the input of~$\bar{\mu}$. $\bar\kappa$ and~$\bar\mu$ share a \emph{common output} alphabet~$\Scal$.}
\end{figure}
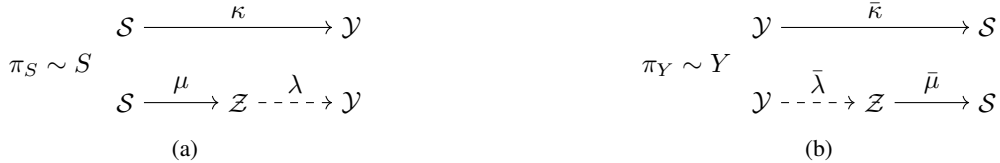

Depending on whether such a randomization is applied at the output or the input, two different ways of quantifying the unique information arise. Utilizing tools from statistical decision theory,~\cite{e16042161} defined the function~$\TUI$ based on the idea of approximating one channel from the other by a randomization at the \emph{output} (see Fig.~\ref{fig:1}). \cite{HarderSalgePolani2013:Bivariate_redundancy} defined the function~$\TSI$ as a difference of two Kullback-Leibler divergence terms where one of the terms implicitly uses a randomization at the \emph{input} (see Fig.~\ref{fig:2}). 
In both cases, the induced decompositions of the total mutual information are nonnegative and satisfy equations~\eqref{eq:MIdec1}-\eqref{eq:MIdec3}. While quantitative differences between the two decompositions have been studied earlier (see, e.g.,~\cite[Figure~1]{e16042161}), the aforementioned distinction seems to have largely gone unnoticed in the literature. Also, as Example~\ref{ex:sktoy} seem to suggest, the unique information is an interesting quantity in its own right that might play a role in bounding the secret key rate. An objective study relating the two quantities is missing.  

In this paper, we bridge these little gaps and make the following contributions:
\begin{itemize}
	\item Given two channels that convey information about the same random variable, we propose two measures of unique information of one channel w.r.t. the other. They are both based on a generalized version of Le Cam's notion of \emph{weighted deficiency}~\cite{lecam,torgersen,raginsky2011} of one channel w.r.t. another. 
	Weighted deficiencies measure the cost of approximating one channel from the other via randomizations. 
	Depending on whether the randomization is applied at the output or the input, two different notions of weighted deficiency arise. We call the respective quantities the \emph{weighted output} and \emph{weighted input deficiencies}. 	
	The new quantities induce nonnegative decompositions of the mutual information. Interestingly, the decomposition induced by the weighted input deficiency coincides with the one proposed in~\cite{HarderSalgePolani2013:Bivariate_redundancy}.
	\item We show that the definition of the unique information proposed in~\cite{e16042161} shares some intuitive and basic properties of the secret key rate~\cite{maurerintrinsic}.
	We give an operational interpretation of this quantity in terms of an upper bound on the \emph{one-way secret key rate}. Theorem~\ref{thm:uppbound} is our main result in this part.
	As a minor side note, for secret key agreement against active adversaries, we restate Maurer's impossibility result~\cite{maurersimul1} in terms of vanishing unique informations in Theorem~\ref{thm:simul}. 
\end{itemize}
Proofs are collected in Appendix~\ref{app:proofs}.

\section{Unique informations and Le Cam Deficiencies}
Suppose that an agent has a finite set of actions~$\Acal$. Each action~$a\in\Acal$ incurs a bounded loss~${\ell}(s,a)$ that depends on the chosen action~$a$ and the state~$s\in\Scal$ of a finite random variable~$S$. Let~$\pi_S$ encode the agents' uncertainty about the true state~$s$. Then, the triple~$(\pi_S,\Acal,\ell)$ is called a \emph{decision problem}. In the sequel, we assume that~$\pi_S$ has full support. Before choosing her action, the agent is allowed to observe a finite random variable~$Z$ through a \emph{channel} from~$\Scal$ to~$\Zcal$ which is a family~$\mu=\{\mu_s\}_{s\in\Scal}$ of probability distributions on~$\Zcal$, one for each possible input~$s\in\Scal$. Let~$\mathsf{M}(\Scal;\Zcal)$ denote the space of all channels from~$\Scal$ to~$\Zcal$ which is the set of all (row) stochastic matrices~$[0,1]^{\Scal\times\Zcal}$. The goal of a rational agent is to choose a strategy~$\rho\in\mathsf{M}(\Zcal;\Acal)$ that minimizes her expected loss or \emph{risk}
\begin{align}
R(\pi_S,\mu,\rho,\ell):=\sum_{s\in\Scal}{\pi_S(s)}\sum_{a\in\Acal}{\rho\circ\mu_s(a)\ell(s,a)},
\end{align}
where~$\rho\circ\mu$ (read~$\rho$ after~$\mu$) denotes the composition of the channels~$\rho$ and~$\mu$.
Writing~$\Acal_\mu=\{\rho\circ\mu\,:\,\rho\in\mathsf{M}(\Zcal;\Acal)\}$, the \emph{optimal risk} when using the channel~$\mu$ is
\begin{align}
R(\pi_S,\mu,\ell):=\min_{\sigma\in\Acal_\mu}\sum_{s\in\Scal}{\pi_S(s)}\sum_{a\in\Acal}{\sigma_s(a)\ell(s,a)}.
\end{align} 
In this minimum, there always exist deterministic optimal strategies. So it suffices to consider deterministic strategies.

Suppose now that the agent is allowed to observe another finite random variable~$Y$ through a second channel~$\kappa\in\mathsf{M}(\Scal;\Ycal)$ with the same input alphabet~$\Scal$. 
When will she \emph{always} prefer~$Z$ to~$Y$? She can rank the variables by comparing her optimal risks: she will \emph{always} prefer~$Z$ over~$Y$ if her optimal risk when using~$Z$ is at most that when using~$Y$ for any decision problem. We have the following definition.
\begin{definition}\label{def:preorder1}
	Let~$(S,Y,Z)\sim P$, $S\sim\pi_S$ and~$\kappa\in\mathsf{M}(\Scal;\Ycal)$, $\mu\in\mathsf{M}(\Scal;\Zcal)$ be two channels with the same input alphabet such that~$P_{SZ}(s,z) = \pi_S(s)\mu_s(z)$ and $P_{SY}(s,y) = \pi_S(s)\kappa_s(y)$.
		We say that~$Z$ is \emph{always more informative} about~$S$ than~$Y$ and write~$Z\uge_{S} Y$ if~$R(\pi_S,\kappa,\ell)\ge R(\pi_S,\mu,\ell)$ for any~$(\pi_S,\Acal,\ell)$.
\end{definition}
She can also rank the variables purely probabilistically: she will \emph{always} prefer~$Z$ over~$Y$ if, knowing~$Z$, she can simulate a single use of~$Y$ by randomly sampling a $y'\in\Ycal$ after each observation~$z\in\Zcal$.
\begin{definition}\label{def:preorder2}
	Write~$Z\mge_{S} Y$ if there exists a random variable~$Y'$ such that the pairs~$(S,Y)$ and $(S,Y')$ are statistically indistinguishable, and $S-Z-Y'$ is a Markov chain.
\end{definition}
The relation~$Z\mge_{S} Y$ is also called the \emph{degradation preorder}. Intuitively,~$Z$ knows everything that $Y$ knows about~$S$ in both these situations. In a classic result, Blackwell showed the equivalence of these two relations. 
\begin{theorem}\label{thm:BW53} (Blackwell's Theorem~\cite{Blackwell1953,BlackwellISIT})
		{$Z\mge_{S} Y$} $\iff$ {$Z\uge_{S} Y$}.	
\end{theorem}
Theorem~\ref{thm:BW53} is a version of the Blackwell's theorem for random variables. 
We say~$\mu$ is \emph{Blackwell sufficient} for~$\kappa$ and write~$\mu\uge_{\Scal} \kappa$ if~$\kappa=\lambda\circ\mu$ for some~$\lambda\in\mathsf{M}(\Zcal;\Ycal)$. 
If~$\pi_S$ has full support, then~$\mu\uge_{\Scal} \kappa\iff Z\uge_{S} Y$~\cite[Theorem 4]{BlackwellISIT}.
In this setting, we are motivated to make the following definition.

\begin{definition}\label{def:UI0} 
	$Y$~\emph{has no unique information} about~$S$ w.r.t.~$Z$ $\logeq$~$Z\mge_{S} Y$.
\end{definition}

Definition~\ref{def:UI0} gives an operational idea when the unique information vanishes~\cite{e16042161}. The converse to the Blackwell's theorem states that if the relation~$Z\mge_{S} Y$ (resp.,~$Y\mge_{S} Z$) does not hold, then there exists a loss function and a set of actions that renders~$Y$ (resp.,~$Z$) more useful. This statement motivates the following definition~\cite{e16042161}.
\begin{definition}\label{def:UIg0}
	$Y$~\emph{has unique information} about~$S$ w.r.t.~$Z$ if there exists a set of actions~$\Acal$ and a loss function~$\ell(s,a)\in\Rb^{\Scal\times\Acal}$ such that~$R(\pi_S,\kappa,\ell)<R(\pi_S,\mu,\ell)$.
\end{definition}

The relation~$\mge_S$ is a preorder on observed variables. In general, one cannot expect two random variables to be comparable, i.e., one can always be simulated by a randomization of the other. On the contrary, most random variables are uncomparable. 
Lucien Le Cam introduced the notion of \emph{deficiencies}~\cite{lecam} and considerably augmented the scope of the Blackwell ordering. Deficiencies measure the cost of approximating one observed variable from the other (and vice versa) via Markov kernels. Maxim Raginsky~\cite{raginsky2011} introduced a broad class of deficiency-like quantities using the notion of a generalized divergence between probability distributions that satisfies a data processing inequality. 
In a spirit similar to~\cite{raginsky2011} and~\cite[Section 6.2]{torgersen}, when the distribution of the common input to the channels is fixed, one can define a \emph{weighted deficiency}. 
\begin{definition}\label{def:gdefo}
	The \emph{weighted output deficiency of~$\mu$ w.r.t.~$\kappa$} is
		\begin{align}\label{eq:weighted_gdefo} 
		\delta_o^{\pi}(\mu,\kappa):=\min_{\lambda\in\mathsf{M}(\Zcal;\Ycal)}D(\kappa\|\lambda\circ\mu|\pi_S),
		\end{align}
		where~$D$ is the Kullback-Leibler divergence and the subscript~$o$ in~$\delta_o^{\pi}$ emphasizes the fact that the randomization is at the \emph{output} of the channel~$\mu$ (see Fig.~\ref{fig:1}). 
\end{definition}
As an immediate consequence, $\delta_o^{\pi}(\mu,\kappa)=0$ if and only if $Z\mge_{S} Y$, which captures the intuition that if~$\delta_o^{\pi}(\mu,\kappa)$ is small, then~$Z$ is \emph{approximately Blackwell sufficient} for~$Y$. 

Le Cam's \emph{randomization criterion}~\cite{lecam} shows that deficiencies quantify the maximal gap in the optimal risks of decision problems when using the channel~$\mu$ rather than~$\kappa$. The next proposition states that bounding the weighted output deficiency is sufficient to ensure that the differences in the optimal risks is also bounded for any decision problem of interest.
\begin{proposition} \label{prop:lecam_suffbound}
	Fix~$\mu\in\mathsf{M}(\Scal;\Zcal)$, $\kappa\in\mathsf{M}(\Scal;\Ycal)$ and a prior probability distribution~$\pi_S$ on~$\Scal$ and write~$\norm{\ell}_\infty=\max_{s,a}\ell(s,a)$. For every~$\epsilon > 0$, if~$\delta_o^{\pi}(\mu,\kappa)\le\epsilon$, then~$R(\pi_S,\mu,\ell)-R(\pi_S,\kappa,\ell)\le\sqrt{\epsilon\tfrac{\lnn{2}}{2}}\norm{\ell}_\infty$ for any set of actions~$\Acal$ and any bounded loss function~$\ell$.
\end{proposition}

Another ordering that has been studied recently is the input-degradedness preorder~\cite{nasser2017} based on whether one channel can be simulated from the other by randomization at the input. 
\begin{definition}\label{def:preorder3}
	Let $\bar{\kappa}\in\mathsf{M}(\Ycal;\Scal)$, $\bar{\mu}\in\mathsf{M}(\Zcal;\Scal)$ be two channels with a common output alphabet. We say that \emph{$\bar{\kappa}$ is input-degraded from~$\bar{\mu}$} and write~$\bar{\mu}\succeq_{\Scal} \bar{\kappa}$ if~$\bar{\kappa}=\bar{\mu}\circ\bar{\lambda}$ for some~$\bar{\lambda}\in\mathsf{M}(\Ycal;\Zcal)$.
\end{definition}
\cite{nasser2017} gave a characterization of input-degradedness that is similar to Blackwell's theorem. The weighted deficiency counterpart of Definition~\ref{def:gdefo} is as follows.
\begin{definition}\label{def:gdefi}
	The \emph{weighted input deficiency of~$\bar{\mu}$ w.r.t.~$\bar{\kappa}$ is}
		\begin{align}\label{eq:weighted_gdefi}
		\delta_i^{\pi}(\bar{\mu},\bar{\kappa}):=\min_{\bar{\lambda}\in\mathsf{M}(\Ycal;\Zcal)}D(\bar{\kappa}\|\bar{\mu}\circ\bar{\lambda}|\pi_Y),
		\end{align}
		where the subscript~$i$ in~$\delta_i^{\pi}$ emphasizes the fact that the randomization is at the \emph{input} of the channel~$\bar{\mu}$ (see Fig.~\ref{fig:2}).
\end{definition}

\vspace{.5cm}
\section{Nonnegative mutual information decompositions}
We propose two nonnegative decompositions of the mutual information between the pair~$(Y,Z)$ and~$S$ based on Definition~\ref{def:gdefo} and Definition~\ref{def:gdefi} of the weighted output and input deficiencies. 
\subsection{Nonnegative decomposition based on weighted output deficiencies}
Consider the following functions on the simplex~$\mathbb{P}_{\Scal\times\Ycal\times\Zcal}$.
\begin{definition}\label{def:decomp_gdefo}
	Let~$(S,Y,Z)\sim P$.
	\begin{subequations}
		\label{subeq:decomp_gdefo}
		\begin{align}
		UI_o(S;Y\backslash Z) &= \max\{\delta_o^{\pi}({\mu},{\kappa}),\delta_o^{\pi}({\kappa},{\mu})+I(S;Y)-I(S;Z)\},\label{subeq:UI_oy}\\
		UI_o(S;Z\backslash Y) &= \max\{\delta_o^{\pi}({\kappa},{\mu}),\delta_o^{\pi}({\mu},{\kappa})+I(S;Z)-I(S;Y)\},\label{subeq:UI_oz}\\
		SI_o(S;Y,Z) &=\min\{I(S;Y)-\delta_o^{\pi}(\mu,\kappa), I(S;Z)-\delta_o^{\pi}(\kappa,\mu)\},\label{subeq:SI_o}\\
		CI_o(S;Y,Z) &= \min\{I(S;Y|Z)-\delta_o^{\pi}(\mu,\kappa), I(S;Z|Y)-\delta_o^{\pi}(\kappa,\mu)\}.\label{subeq:CI_o}
		\end{align}
	\end{subequations}
\end{definition}
\begin{remark}
	The functions~$UI_o$ and~$SI_o$ depend only on the triple~$(\pi,{\kappa},{\mu})$. The function~$CI_o$ depends on the full joint~$P$.
\end{remark}
It is easy to check that the functions~\eqref{subeq:decomp_gdefo} satisfy the information decomposition equations~\eqref{eq:MIdec1}-\eqref{eq:MIdec3}. 
\begin{proposition} [Nonnegativity] \label{lem:positivity_gdefo}
	$SI_o$, $UI_o$ and $CI_o$ are nonnegative functions.
\end{proposition}
\begin{lemma}
	\label{lem:UIo-zero}
	$UI_{o}(S;Y\backslash Z)$ vanishes if and only if $Y$ has no unique information about~$S$ w.r.t.~$Z$ (according to
	Definition~\ref{def:UI0}).
\end{lemma}

\subsection{Nonnegative decomposition based on weighted input deficiencies}
Consider the following functions on the simplex~$\mathbb{P}_{\Scal\times\Ycal\times\Zcal}$.
\begin{definition}\label{def:decomp_gdefi}
	Let~$(S,Y,Z)\sim P$.
	\begin{subequations}
		\label{subeq:decomp_gdefi}
		\begin{align}
		UI_i(S;Y\backslash Z) &= \max\{\delta_i^{\pi}(\bar{\mu},\bar{\kappa}),\delta_i^{\pi}(\bar{\kappa},\bar{\mu})+I(S;Y)-I(S;Z)\},\label{subeq:UI_iy}\\
		UI_i(S;Z\backslash Y) &= \max\{\delta_i^{\pi}(\bar{\kappa},\bar{\mu}),\delta_i^{\pi}(\bar{\mu},\bar{\kappa})+I(S;Z)-I(S;Y)\},\label{subeq:UI_iz}\\
		SI_i(S;Y,Z) &=\min\{I(S;Y)-\delta_i^{\pi}(\bar{\mu},\bar{\kappa}),I(S;Z)-\delta_i^{\pi}(\bar{\kappa},\bar{\mu})\},\label{subeq:SI_i}\\
		CI_i(S;Y,Z) &= \min\{I(S;Y|Z)-\delta_i^{\pi}(\bar{\mu},\bar{\kappa}),I(S;Z|Y)-\delta_i^{\pi}(\bar{\kappa},\bar{\mu})\}.\label{subeq:CI_i}
		\end{align}
	\end{subequations}
\end{definition}
\begin{remark}
	The functions~$UI_i$ and~$SI_i$ depend only on the tuple~$(\pi_Y,\pi_Z,\bar{\kappa},\bar{\mu})$. $CI_i$ depends on the full joint~$P$.
\end{remark}
It is easy to see that the functions~\eqref{subeq:decomp_gdefo} satisfy the information decomposition equations~\eqref{eq:MIdec1}-\eqref{eq:MIdec3}. 
\begin{proposition} [Nonnegativity]\label{lem:positivity_gdefi}
	$SI_i$, $UI_i$ and $CI_i$ are nonnegative functions.
\end{proposition}

\addtocounter{theorem}{1}

\cite{HarderSalgePolani2013:Bivariate_redundancy} defined a measure of \emph{shared information} based on reverse information projections~\cite{csiszarIproj} to a convex set of probability measures.
\begin{definition} \label{def:SIred}
	    For~$C\subset\mathbb{P}_{\Scal}$, let~$\conv(C)$ denote the convex hull of~$C$. 
		Let $$Q_{y\searrow Z}(S)\in \argmin_{Q\in\conv\left(\{\bar{\mu}_z\}_{z\in\Zcal}\right)\subset\mathbb{P}_{\Scal}} D(\bar{\kappa}_y\|Q)$$ be the \emph{reverse I-projection} of~$\bar{\kappa}_y$ onto the convex hull of the points~$\{\bar{\mu}_z\}_{z\in\Zcal}\in\mathbb{P}_{\Scal}$. Define the \emph{projected information of~$Y$ onto~$Z$ w.r.t. S} as
		\begin{align}
		I_{S}(Y\searrow Z):=\mathbb{E}_{(s,y)\sim\bar{\kappa}\times\pi_Y} \log \tfrac{Q_{y\searrow Z}(s)}{\bar{\kappa}\circ\pi_Y(s)},
		\end{align} and the shared information
		\begin{align}
		SI_{red}(S;Y,Z):=\min\{I_{S}(Y\searrow Z),I_{S}(Z\searrow Y)\} \label{eq:SI_red}.
		\end{align}
\end{definition}
For an account of some intuitive properties of the function~$SI_{red}$ as a measure of shared information, see~\cite[Section II.B]{HarderSalgePolani2013:Bivariate_redundancy} and~\cite{WilliamsBeer}.

Proposition~\ref{prop:UIred_equals_gdefi} states that implicit in the above construction is the weighted input deficiency~$\delta_i^{\pi}(\bar{\mu},\bar{\kappa})$. 
\begin{proposition} \label{prop:UIred_equals_gdefi}
	$I_{S}(Y\searrow Z)=I(S;Y)-\delta_i^{\pi}(\bar{\mu},\bar{\kappa})$.
\end{proposition}
An immediate consequence of Proposition~\ref{prop:UIred_equals_gdefi} is that the decomposition proposed in~\cite{HarderSalgePolani2013:Bivariate_redundancy} and that in Definition~\ref{def:decomp_gdefi} are equivalent. 
\begin{proposition}\label{thm:equivalence_Ired_gdefi}
  \addtocounter{equation}{1}
$SI_{red}=SI_i$, $UI_{red}=UI_i$, $CI_{red}=CI_i$,
where~$UI_{red}$ and~$CI_{red}$ are the corresponding unique and complementary informations derived from~\eqref{eq:SI_red} and satisfying equations~\eqref{eq:MIdec1}-\eqref{eq:MIdec3}.
\end{proposition}

\vspace{.5cm}
\section{Minimum-synergy unique information}
\cite{e16042161} proposed a nonnegative decomposition of the mutual information based on the idea that the unique and shared information should depend only on the marginal distributions of the pairs~$(S,Y)$ and~$(S,Z)$. 
\begin{definition}\label{def:decomp_minsyn}
	Let~$(S,Y,Z)\sim P$ and let~${\kappa}\in\mathsf{M}(\Scal;\Ycal)$, ${\mu}\in\mathsf{M}(\Scal;\Zcal)$ be two channels with the same input alphabet such that~$P_{SY}(s,y) = \pi_S(s){\kappa}_s(y)$ and~$P_{SZ}(s,z) = \pi_S(s){\mu}_s(z)$. Define 
		\begin{subequations}
			\label{subeq:decomp_minsyn}
			\begin{align}
			\Delta_P = \big\{Q \in \mathbb{P}_{\Scal\times\Ycal\times\Zcal}\colon &Q_{SY}(s,y)=\pi_S(s)\kappa_s(y),\notag\\
			&Q_{SZ}(s,z)=\pi_S(s)\mu_s(z)\big\},\label{subeq:delP}\\
			UI(S;Y\backslash Z) &= \min_{Q \in \Delta_P} I_Q(S;Y|Z),\label{subeq:UIy}\\
			UI(S;Z\backslash Y) &= \min_{Q \in \Delta_P} I_Q(S;Z|Y),\label{subeq:UIz}\\
			SI(S;Y,Z)&=\max_{Q \in \Delta_P} CoI_Q(S;Y;Z),\label{subeq:SI}\\
			CI(S;Y,Z)&= I(S;Y|Z)-UI(S;Y\backslash Z),\label{subeq:CI}
			\end{align}
		\end{subequations}
	where~$CoI$ is the coinformation and the subscript~$Q$ in~$CoI_Q$ and~$I_Q$ denotes that joint distribution on which the quantities are computed.
\end{definition}
In Appendix~\ref{app:opt}, we briefly comment on the optimization problems in Definitions~\ref{def:decomp_minsyn}, \ref{def:gdefi} and~\ref{def:gdefo}.
\begin{remark}
	The functions~$UI$ and~$SI$ depend only on the triple~$(\pi,{\kappa},{\mu})$. The function~$CI$ depends on the full joint~$P$. 
\end{remark}
\begin{lemma} [{{\cite[Lemma~6]{e16042161}}}]
	\label{lem:UI-zero}
	$UI(S;Y\backslash Z)$ vanishes if and only if $Y$ has no unique information about~$S$ w.r.t.~$Z$ (according to
	Definition~\ref{def:UI0}).
\end{lemma}
\begin{remark} The following trivial bounds follow from~\eqref{eq:MIdec1}-\eqref{eq:MIdec3}.
	\begin{align}
	I(S;Y)-I(S;Z) \le UI(S;Y\backslash Z) \le \min\{I(S;Y),I(S;Y|Z)\}.
	\end{align}
	These bounds are also valid for the functions~$UI_o$ and~$UI_i$.
	In the adversarial setting in Example~\ref{ex:sktoy}, if either Eve has less information about~$S$ than Bob or, by symmetry, less information about~$Y$ than Alice, then Alice and Bob can exploit this difference to extract a secret key. In such a setting, bounds on the \emph{unique information common to~$S$ and~$Y$ w.r.t.~$Z$} are useful.
		\begin{align}
		\max\{I(S;Y)-I(S;Z),I(Y;S)-I(Y;Z)\}\le\max\{UI(S;Y\backslash Z),UI(Y;S\backslash Z)\}\le \min\{I(S;Y),I(S;Y|Z)\}.
		\end{align}	
		An interesting observation is that these bounds match the trivial bounds on the two-way secret key rate~\cite{maurerintrinsic} (see Section~\ref{subsec:skrate}).
		\begin{align}
		\max\{I(S;Y)-I(S;Z),I(Y;S)-I(Y;Z)\}\le \SKK{S}{Y}{Z} \le \min\{I(S;Y),I(S;Y|Z)\}.
		\end{align}
\end{remark}

The following lemma states that the quantities~$UI$, $SI$ and~$CI$ in Definition~\ref{def:decomp_minsyn} bound the unique, shared and complementary components in any nonnegative decomposition of the mutual information under an assumption that is in keeping with the Blackwell ordering.
\begin{lemma}[{{\cite[Lemma~3]{e16042161}}}]\label{lem:minsyn}
	Let $\TUI(S;Y\backslash Z)$, $\TUI(S;Z\backslash Y)$, $\TSI(S;Y,Z)$ and $\TCI(S;Y,Z)$ be nonnegative functions on~$\mathbb{P}_{\Scal\times\Ycal\times\Zcal}$
		satisfying equations~\eqref{eq:MIdec1}-\eqref{eq:MIdec3} and assume that the following holds:
		\begin{enumerate}
			\item[$(\ast)$] $\TUI$ depends only on the triple~$(\pi,{\kappa},{\mu})$.
		\end{enumerate}
		Then $\TUI \le UI$, $\TSI\ge SI$ and $\TCI \ge CI$ with equality if and only if there exists~$Q\in\Delta_P$ such that~$\TCI_{Q}(S:Y;Z)=0$.
\end{lemma}
Lemma~\ref{lem:minsyn} is consistent with our interpretation of the function~$UI$ as the \emph{minimum-synergy unique information}.
\begin{corollary}
  \begin{align*}
	\delta_o^{\pi}(\mu,\kappa) &\le UI_o(S;Y\backslash Z) \le UI(S;Y\backslash Z),
          \\
          \delta_i^{\pi}(\bar{\mu},\bar{\kappa}) &\le UI_i(S;Y\backslash Z) \le UI(S;Y\backslash Z).
	\end{align*}
\end{corollary}

Proposition~\ref{prop:vanishingUIequiv} follows from Lemmas~\ref{lem:UIo-zero} and \ref{lem:UI-zero}, and Definition \ref{def:gdefo}.
\begin{proposition} \label{prop:vanishingUIequiv}
	$\delta_o^{\pi}(\mu,\kappa)=0 \iff UI_o(S;Y\backslash Z)=0 \iff UI(S;Y\backslash Z)=0$.
\end{proposition}

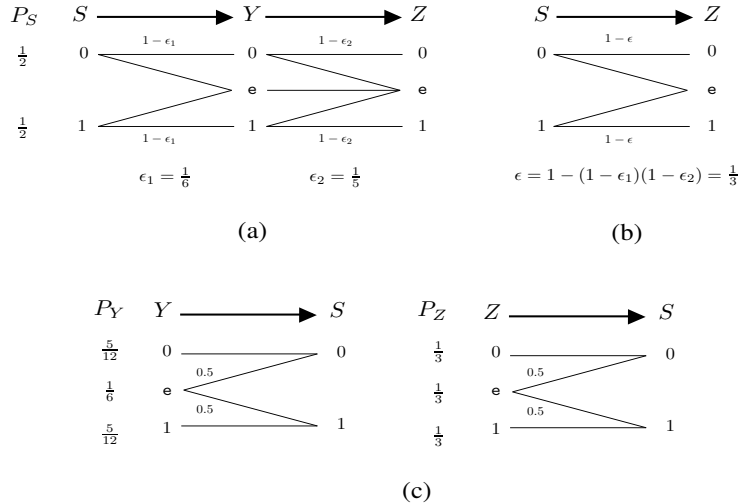
\begin{figure}
\centering
\tikzset{every picture/.style={line width=0.75pt}} 
\resizebox {10cm} {6.8cm} {
	\begin{tikzpicture}[x=0.75pt,y=0.75pt,yscale=-1,xscale=1]

	
	\draw  [line width=0.1mm]  (177.07,1093.35) -- (245.77,1093.35) ;

	\draw  [line width=0.1mm]  (177.07,1093.35) -- (244.55,1073.18) ;

	\draw  [line width=0.1mm]  (177.07,1053.17) -- (244.55,1073.18) ;

	\draw  [line width=0.1mm]  (177.07,1053.17) -- (245.77,1053.17) ;

	\draw  [line width=0.1mm]  (262.07,1093.35) -- (330.77,1093.35) ;

	\draw  [line width=0.1mm]  (262.07,1093.35) -- (329.55,1073.18) ;

	\draw  [line width=0.1mm]  (262.07,1053.17) -- (329.55,1073.18) ;

	\draw  [line width=0.1mm]  (262.07,1053.17) -- (330.77,1053.17) ;

	\draw  [line width=0.1mm]  (262.4,1073.17) -- (329.55,1073.18) ;

	\draw  [line width=0.3mm]  (176.5,1032.23) -- (243.5,1032.23) ;
	\draw [shift={(245.5,1032.54)}, rotate = 180.26] [fill={rgb, 255:red, 0; green, 0; blue, 0 }  ][line width=0.75]  [draw opacity=0] (8.93,-4.29) -- (0,0) -- (8.93,4.29) -- cycle    ;
	
	\draw [line width=0.3mm]   (261.5,1032.23) -- (328.5,1032.23) ;
	\draw [shift={(330.5,1032.54)}, rotate = 180.26] [fill={rgb, 255:red, 0; green, 0; blue, 0 }  ][line width=0.75]  [draw opacity=0] (8.93,-4.29) -- (0,0) -- (8.93,4.29) -- cycle    ;
	
	\draw [line width=0.1mm]   (407.07,1093.35) -- (475.77,1093.35) ;

	\draw [line width=0.1mm]   (407.07,1093.35) -- (474.55,1073.18) ;

	\draw [line width=0.1mm]   (407.07,1053.17) -- (474.55,1073.18) ;

	\draw [line width=0.1mm]   (407.07,1053.17) -- (475.77,1053.17) ;

	\draw [line width=0.3mm]   (408.5,1032.23) -- (475.5,1032.23) ;
	\draw [line width=0.1mm] [shift={(477.5,1032.54)}, rotate = 180.26] [fill={rgb, 255:red, 0; green, 0; blue, 0 }  ][line width=0.75]  [draw opacity=0] (8.93,-4.29) -- (0,0) -- (8.93,4.29) -- cycle    ;
	
	\draw  [line width=0.1mm]  (287.77,1220.17) -- (219.07,1220.17) ;

	\draw  [line width=0.1mm]  (287.77,1220.17) -- (220.29,1240.35) ;

	\draw  [line width=0.1mm]  (287.77,1260.35) -- (220.29,1240.35) ;

	\draw  [line width=0.1mm]  (287.77,1260.35) -- (219.07,1260.35) ;

	\draw [line width=0.3mm]   (218.5,1198.23) -- (285.5,1198.23) ;
	\draw [shift={(287.5,1198.54)}, rotate = 180.26] [fill={rgb, 255:red, 0; green, 0; blue, 0 }  ][line width=0.75]  [draw opacity=0] (8.93,-4.29) -- (0,0) -- (8.93,4.29) -- cycle    ;
	
	\draw  [line width=0.1mm]  (453.77,1221.17) -- (385.07,1221.17) ;

	\draw  [line width=0.1mm]  (453.77,1221.17) -- (386.29,1241.35) ;

	\draw  [line width=0.1mm]  (453.77,1261.35) -- (386.29,1241.35) ;

	\draw  [line width=0.1mm]  (453.77,1261.35) -- (385.07,1261.35) ;

	\draw  [line width=0.3mm]  (384.5,1199.23) -- (451.5,1199.23) ;
	\draw [line width=0.1mm] [shift={(453.5,1199.54)}, rotate = 180.26] [fill={rgb, 255:red, 0; green, 0; blue, 0 }  ][line width=0.75]  [draw opacity=0] (8.93,-4.29) -- (0,0) -- (8.93,4.29) -- cycle    ;

	\draw (140,1032.08) node [scale=0.9]  {${\textstyle P_{S}}$};
	\draw (139,1054.08) node [scale=0.7]  {$\tfrac{1}{2}$};
	\draw (139,1094.08) node [scale=0.7]  {$\tfrac{1}{2}$};
	\draw (168,1030.57) node [scale=0.9]  {${\textstyle S}$};
	\draw (255,1030.57) node [scale=0.9]  {$Y$};
	\draw (339,1030.57) node [scale=0.9]  {$Z$};
	\draw (170,1052.08) node [scale=0.7]  {$0$};
	\draw (169,1093.08) node [scale=0.7]  {$1$};
	\draw (255,1052.08) node [scale=0.7]  {$0$};
	\draw (255,1093.08) node [scale=0.7]  {$1$};
	\draw (255,1073.08) node [scale=0.7]  {$\mathtt{e}$};
	\draw (341,1052.08) node [scale=0.7]  {$0$};
	\draw (341,1093.08) node [scale=0.7]  {$1$};
	\draw (341,1073.08) node [scale=0.7]  {$\mathtt{e}$};
	\draw (208,1045.62) node [scale=0.5]  {$1-\epsilon _{1}$};
	\draw (297,1045.62) node [scale=0.5]  {$1-\epsilon _{2}$};
	\draw (208,1100.62) node [scale=0.5]  {$1-\epsilon _{1}$};
	\draw (297,1100.62) node [scale=0.5]  {$1-\epsilon _{2}$};
	\draw (440,1044.62) node [scale=0.5]  {$1-\epsilon $};
	\draw (401,1030.57) node [scale=0.9]  {${\textstyle S}$};
	\draw (487,1030.57) node [scale=0.9]  {$Z$};
	\draw (401,1093.08) node [scale=0.7]  {$1$};
	\draw (401,1052.08) node [scale=0.7]  {$0$};
	\draw (487,1051.08) node [scale=0.7]  {$0$};
	\draw (440,1100.62) node [scale=0.5]  {$1-\epsilon $};
	\draw (487,1093.08) node [scale=0.7]  {$1$};
	\draw (487,1073.08) node [scale=0.7]  {$\mathtt{e}$};
	\draw (212,1241.08) node [scale=0.7]  {$\mathtt{e}$};
	\draw (212,1261.08) node [scale=0.7]  {$1$};
	\draw (212,1218.08) node [scale=0.7]  {$0$};
	\draw (444,1120.62) node [scale=0.7]  {$\epsilon =1-( 1-\epsilon _{1})( 1-\epsilon _{2}) =\tfrac{1}{3}$};
	\draw (300,1219.08) node [scale=0.7]  {$0$};
	\draw (300,1259.08) node [scale=0.7]  {$1$};
	\draw (183,1195.08) node [scale=0.9]  {${\textstyle P_{Y}}$};
	\draw (183,1218.08) node [scale=0.7]  {$\tfrac{5}{12}$};
	\draw (183,1265.08) node [scale=0.7]  {$\tfrac{5}{12}$};
	\draw (210,1195.57) node [scale=0.9]  {$Y$};
	\draw (298,1195.57) node [scale=0.9]  {${\textstyle S}$};
	\draw (231,1230.08) node [scale=0.5]  {$0.5$};
	\draw (231,1251.08) node [scale=0.5]  {$0.5$};
	\draw (183,1241.08) node [scale=0.7]  {$\tfrac{1}{6}$};
	\draw (378,1241.08) node [scale=0.7]  {$\mathtt{e}$};
	\draw (378,1262.08) node [scale=0.7]  {$1$};
	\draw (378,1219.08) node [scale=0.7]  {$0$};
	\draw (466,1220.08) node [scale=0.7]  {$0$};
	\draw (466,1260.08) node [scale=0.7]  {$1$};
	\draw (346,1196.08) node [scale=0.9]  {${\textstyle P_{Z}}$};
	\draw (349,1219.08) node [scale=0.7]  {$\tfrac{1}{3}$};
	\draw (349,1266.08) node [scale=0.7]  {$\tfrac{1}{3}$};
	\draw (376,1196.57) node [scale=0.9]  {$Z$};
	\draw (464,1196.57) node [scale=0.9]  {${\textstyle S}$};
	\draw (398,1230.08) node [scale=0.5]  {$0.5$};
	\draw (398,1252.08) node [scale=0.5]  {$0.5$};
	\draw (349,1242.08) node [scale=0.7]  {$\tfrac{1}{3}$};
	\draw (211,1121.62) node [scale=0.7]  {$\epsilon _{1} =\tfrac{1}{6}$};
	\draw (297,1121.62) node [scale=0.7]  {$\epsilon _{2} =\tfrac{1}{5}$};
	
	\draw (255,1152.34) node  [align=left] {(a)};
	\draw (444,1153.34) node  [align=left] {(b)};
	\draw (338,1297.34) node  [align=left] {(c)};
	
	\end{tikzpicture}}
\caption{\emph{Distribution in Example~\ref{prop:vanishingUIequiv}b):} (a) Concatenated erasure channels with a binary symmetric input distribution. (b) Equivalent erasure channel $P_{Z|S}$ with erasure probability $\epsilon=\tfrac{1}{3}$. (c) The ``reverse'' erasure channels $P_{S|Y}$ and $P_{S|Z}$.}
\label{fig:counterexample}
\end{figure}

In~\cite{e16042161}, it was claimed that the vanishing sets of $UI_{\text{red}}=UI_{i}$ and $UI$ agree.  In the published version of this paper~\cite{Allerton2018}, this was used to show that the vanishing sets of $UI_{i}$ and $\delta^{\pi}_{i}$ agree with that of~$UI$.  As the following example shows, this is not correct:
\begin{example*}
	Consider the distribution depicted in Figure~\ref{fig:counterexample}(a). $P_S$ is a binary symmetric distribution and $P_{Y|S}$ and $P_{Z|Y}$ are symmetric erasure channels with erasure probabilities $\epsilon_1=\tfrac{1}{6}$ and $\epsilon_2=\tfrac{1}{5}$, resp. We have:
	\begin{itemize}
		\item $S-Y-Z$ is a Markov chain by construction. The erasure probability of the induced erasure channel $P_{Z|S}$ shown in Figure~\ref{fig:counterexample}(b) is greater than that of $P_{Y|S}$, whence $UI(S;Y\backslash Z)=I(S;Y|Z)=\tfrac{1}{6}>0$. 
		\item On the other hand, the induced ``reverse'' erasure channels $P_{S|Y}$ and $P_{S|Z}$ are identical (see Figure~\ref{fig:counterexample}(c)). Thus, $\delta_i^{\pi}=0$.
	\end{itemize}
\end{example*}

\vspace{.3cm}
\subsection{\texorpdfstring{$UI$}{UI} is an upper bound on the one-way secret key rate} \label{subsec:skrate}
In this section, we show that the function~$UI$ has a meaningful operational interpretation in a task where the goal is to extract a secret key from shared randomness and public communication.

In a \emph{two-way} secret key agreement protocol for the \emph{source model}~\cite{maurer1993,ahlswede1993,maurerintrinsic}, Alice, Bob and Eve observe~$n$ i.i.d. copies of random variables~$S$, $Y$ and~$Z$ respectively, where~$(S,Y,Z)$ is distributed according to some joint distribution~$P$ assumed to be known to all parties. The protocol proceeds in rounds: In each round either Alice or Bob can transmit a message over an insecure but authenticated public discussion channel. If Alice (resp., Bob) transmits message~$C_i$ in round~$i$, then~$C_i$ is a function of~$S^n$ (resp., $Y^n$) and all the messages received so far. After~$r$ rounds, Alice (resp., Bob) computes a key~$K_A^m\in\Kcal^{m}$ for~$\Kcal=\{0,1\}$ (resp.,~$K_B^m\in\Kcal^{m}$) as a function of~$S^n$ (resp.,~$Y^n$) and~$C:=(C_1,C_2,\cdots,C_{r})$, the collection of messages sent over the public channel. The protocol is \emph{one-way} if Alice is allowed to use the public channel only once and Bob cannot transmit at all: Alice computes a key~$K_A^m$ and a message~$C$ for Bob as a function of~$S^n$. Bob computes a key~$K_B^m$ as a function of~$Y^n$ and~$C$. In the limit~$n\to\infty$, the secret key must satisfy the following conditions:
\begin{align}
\Prv{K_A^{m}\ne K_B^{m}}=0,\quad \log|\Kcal^m|-H(K_A^m|Z^nC)=0.
\end{align}
The largest achievable rate~$\lim_{n\rightarrow \infty}\frac {m}n$ at which Alice and Bob can distill a key in the two-way and the one-way communication scenarios are resp. called, the \emph{two-way secret key rate}~$\SKK{S}{Y}{Z}$ and the \emph{one-way secret key rate}~$\SK{S}{Y}{Z}$. 

An exact expression for the one-way secret key rate is known. 
\begin{theorem}[{{\cite[Theorem~1]{ahlswede1993}}}]\label{thm:skaRate}
	The \emph{one-way secret key rate} $\SK{S}{Y}{Z}$ for the source model is
	    \begin{align*}
	    \SK{S}{Y}{Z}=\max\limits_{P_{UV|SYZ}} & I(U;Y|V)-I(U;Z|V) \label{eq:onewayrate}
		\end{align*}
for random variables~$U$, $V$ such that $YZ-S-UV$ is a Markov chain, and where both~$U$ and~$V$ have range of size at most~$|\Scal|+1$.
\end{theorem}

The one-way secret key rate is a lower bound on the two-way secret key rate. 

General properties of upper bounds on the secret key rates have been investigated under the rubric of \emph{protocol monotones}---nonnegative real-valued functionals of joint distributions that can never increase during protocol execution (see e.g., \cite{maurerunbreakable,gohari,christandl,gohari3}). For example, the \emph{intrinsic information}, an upper bound on the two-way secret key rate is a protocol monotone~\cite{RennerW03}.

We show that the function~$UI$ shares some intuitive and basic properties of the secret key rate.
Lemma~\ref{lem:LORR} states that if Alice and Bob discard certain realizations of their random variables by restricting their ranges, then the~$UI$ can never increase. See~\cite[Lemma 3]{maurerintrinsic} for a counterpart of this property for the two-way secret key rate. 
\begin{lemma} [Monotonicity under range restrictions~\cite{ISIT_RBOJ14}] \label{lem:LORR}
	$UI((S,S');(Y,Y') \backslash Z) \geq UI(S;Y \backslash Z)$.
\end{lemma}

The following Lemma states that~$UI$ can never increase under local operations of Alice and Bob. The counterpart of this lemma for the secret key rate is~\cite[Lemma 4]{maurerintrinsic}.
\begin{lemma}[Monotonicity under Local Operations]\label{lem:LO}
	$UI$ cannot increase under local operations of $S$ or~$Y$.
\end{lemma}

Lemma~\ref{lem:LOEve} states that if Eve is allowed access to some additional side information, then the~$UI$ can only decrease. 
See~\cite[Lemma 5]{maurerintrinsic} for a counterpart of this property for the two-way secret key rate. 
\begin{lemma} [Monotonicity under adversarial side information~\cite{ISIT_RBOJ14}] \label{lem:LOEve}
	$UI(S;Y \backslash (Z,Z')) \leq UI(S;Y \backslash Z)$.
\end{lemma}

Suppose Alice publicly announces the value of a random variable. Then Lemma~\ref{lem:PC} states that~$UI$ can never increase. 
\begin{lemma} [Monotonicity under one-way public communication] \label{lem:PC}
	$UI(S;(Y,f(S)) \backslash (Z,f(S))) \leq UI(S;Y \backslash Z)$ for all functions~$f$ over the support of~$S$.
\end{lemma}

The following two properties, additivity and asymptotic continuity are important if the function~$UI$ is to furnish an upper bound on the asymptotic rate of transforming a given joint distribution into a secret key.

Lemma~\ref{lem:AD} states that~$UI$ is additive on tensor products.
\begin{lemma} [Additivity under tensor products.~{{\cite[Lemma~19]{e16042161}}}] 
	\label{lem:AD}
	For independent pairs of jointly distributed random variables~$(S_{1},Y_{1},Z_{1})$ and~$(S_{2},Y_{2},Z_{2})$, 
	\begin{align*}
		UI((S_1,S_2);(Y_1,Y_2)\backslash (Z_1,Z_2) = UI(S_1;Y_1\backslash Z_1) + UI(S_2;Y_2\backslash Z_2).
	\end{align*}
\end{lemma}

We also have asymptotic continuity for the~$UI$.
\begin{theorem}[Asymptotic continuity] \label{thm:AC}
	$UI$ is asymptotically continuous.
\end{theorem}

The following theorem gives sufficient conditions for a function to be an upper bound for the secret key rate.
\begin{theorem} [{{\cite[Theorem 3.1]{christandl}},\cite[Lemma 2.10]{maurerunbreakable}}] \label{thm:monotone}
	Let~$M$ be a nonnegative real-valued function of the joint distribution of the triple~$(S,Y,Z)$ such that the following holds:
		\begin{trivlist}
			\item \emph{1. Local operations (LO) of Alice or Bob cannot increase~$M$:} For all jointly distributed RVs~$(S,Y,Z,S')$ such that~$YZ-S-S'$ is a Markov chain, $M(S,Y,Z) \ge M(S',Y,Z)$ (and likewise for~$Y$).
            \item \emph{2. Public communication (PC) by Alice cannot increase~$M$:} $M\bigl((S,f(S)),(Y,f(S)),(Z,f(S))\bigr) \leq M(S,Y,Z)$ for all functions~$f$ over the support of~$S$. 
            \item \emph{3. Normalization:} For a perfect secret bit $P_{SS\Delta}(0,0,\delta)=P_{SS\Delta}(1,1,\delta)=\tfrac{1}{2}$, $M(S,S,\Delta)=1$. 
			\item \emph{4. Asymptotic continuity:} $M$ is a asymptotically continuous function of~$(S,Y,Z)$. 
			\item \emph{5. Additivity:} $M$ is additive on tensor products.
		\end{trivlist}
               Then~$M$ is an upper bound for the \emph{one-way secret key rate}.
        
        If, in addition, $M$ does not increase under public communication by Bob (property~2., with $f(S)$ replaced by
        $g(Y)$ for some function $g$ over the support of~$Y$), then $M$ is an upper bound for the two-way secret key
        rate.
\end{theorem}

Theorem~\ref{thm:uppbound} is our main result in this Section.
\begin{theorem} \label{thm:uppbound}
   $UI$ is an upper bound for the one-way secret key rate.
\end{theorem}

\subsection{Vanishing unique informations and secret key agreement against active adversaries} \label{subsec:simulatability}
The secret key agreement scenario in Section~\ref{subsec:skrate} assumes that the public discussion channel is authenticated, i.e., Eve is only a passive adversary. When this assumption is no longer valid and Eve has both read and write access to the public channel, an all-or-nothing result is known~\cite{maurersimul1}: Either the same secret key rate can be achieved as in the authentic channel case, or nothing at all. Maurer defined the following property of a distribution to characterize the impossibility of secret key agreement against active adversaries.
\begin{definition} \label{def:simulatability}
	Given~$(S,Y,Z)\sim P$, we say that \emph{$Y$ is simulatable by~$Z$ w.r.t.~$S$} and write~$\simu_S(Z\to Y)$ if there exists a random variable~$Y'$ such that the pairs~$(S,Y)$ and $(S,Y')$ are statistically indistinguishable, and $S-Z-Y'$ is a Markov chain.
\end{definition}  
One would immediately recognize that~$\simu_S(Z\to Y)$ and $Z\mge_{S} Y$ in Definition~\ref{def:preorder1} are equivalent. 
Let~$S_{\leftrightarrow}^{\ast}({S};{Y}|{Z})$ denote the secret key rate in the active adversary scenario.
We restate Maurer's impossibility result (Theorem 11 in~\cite{maurersimul1}) in terms of the function~$UI$.
\begin{theorem} [{{\cite[Theorem 11]{maurersimul1}}}]\label{thm:simul}
	Let~$(S,Y,Z)\sim P$ be a distribution with~$\SKK{S}{Y}{Z}>0$. If either~$UI(S;Y\backslash Z)=0$ or~$UI(Y;S\backslash Z)=0$, then~$S_{\leftrightarrow}^{\ast}({S};{Y}|{Z})=0$, else~$S_{\leftrightarrow}^{\ast}({S};{Y}|{Z})=\SKK{S}{Y}{Z}$.
\end{theorem}
\begin{remark}
Theorem~\ref{thm:simul} gives an operational significance to the vanishing~$UI$, namely, if either~$S$ or~$Y$ possess no unique information about each other w.r.t. $Z$, then Alice and Bob have no advantage in a secret key agreement task against an active Eve.
By Proposition~\ref{prop:vanishingUIequiv}, the same is true for~$UI_o$.
\end{remark}
Example~\ref{ex:sktoy1} shows a distribution for which $\SKK{S}{Y}{Z}>0$ but~$S_{\leftrightarrow}^{\ast}({S};{Y}|{Z})=0$. 
\begin{example}[{{\cite[Example~4]{gisinlinking}}}]\label{ex:sktoy1}
Consider the distribution
\begin{center}
	\resizebox{.2\textwidth}{!}{
	\centering
	\begin{tabular}{clcc}
		\toprule
		\multicolumn{1}{c}{} & \multicolumn{2}{c}{$S$}\\
		\cmidrule{2-3}
		$Y$ ($Z$) & 0 & 1\\
		\midrule
		0 &  $\nicefrac{1}{5}$\ (0) & $\nicefrac{1}{5}$\ (0)\\
		&  $\nicefrac{1}{5}$\ (1) & 0\ (1)\\
		1 &  $\nicefrac{1}{5}$\ (0) & 0\ (0)\\
		& 0\ (1) & $\nicefrac{1}{5}$\ (1)\\
		\bottomrule
	\end{tabular}}
\end{center}
where $Z$'s value is shown in parentheses. This distribution has~$I(S;Y\!\!\downarrow \!Z)=\TSI(S; Y,Z)=0.02$ and~$\TCI(S; Y,Z)=0.55$. \cite{gisinlinking} showed that a secret key agreement protocol exists such that~$\SKK{S}{Y}{Z}>0$. However since the pairwise marginal distributions of~$(S,Y)$, $(S,Z)$ and~$(Y,Z)$ are all identical, all the unique informations vanish. Hence~$S_{\leftrightarrow}^{\ast}({S};{Y}|{Z})=0$.
\end{example}

\section{Conclusion} 
The information decomposition framework extends earlier ideas to define information measures that make it possible to do a finer analysis than is possible with Shannon's mutual information alone. For example, measures of redundancy and synergy have long been sought in the neural sciences~\cite{syncode,LathamNirenberg05:Synergy_and_redundancy_revisited,schneidman2003synergy}. 

In this paper, we proposed two new quantities that can be interpreted as unique informations in the context of nonnegative mutual information decompositions. The quantities are derived using a generalized version of weighted Le Cam deficiencies that have a rich heritage in the theory of comparison of statistical experiments~\cite{torgersen}. 
We related the proposed quantities to the function~$UI$ proposed in~\cite{e16042161}. 
We gave an operational interpretation of the latter in terms of an upper bound on the number of secret key bits extractable per copy of a given joint distribution using local operations and one-way public communication. It might be of independent interest to characterize the set of distributions for which two-way secret key agreement is possible at a rate given by the unique information.

\pagebreak                           
\section*{APPENDIX}\label{appendix}
\subsection{Proofs} \label{app:proofs}
\begin{proof} [Proof of Proposition~\ref{prop:lecam_suffbound}]
	If~$\delta_o^{\pi}(\mu,\kappa)\le\epsilon$, then we can find some~$\lambda\in\mathsf{M}(\Zcal;\Ycal)$ such that~$D(\kappa\|\lambda\circ\mu|\pi_S)\le\epsilon$. Let~$\rho'\in\mathsf{M}(\Ycal;\Acal)$ and let $\rho=\rho'\circ\lambda$. Then
	\begin{align*}
	R(\pi_S,\mu,\rho,\ell)&-R(\pi_S,\kappa,\rho',\ell)\\
	=&\mathbb{E}_{s\sim\pi_S}\left[\mathbb{E}_{a\sim\rho\circ\mu_s}\ell(s,a)-\mathbb{E}_{a\sim\rho'\circ\kappa_s}\ell(s,a)\right]\\
	\le&\mathbb{E}_{s\sim\pi_S}\norm{\rho\circ\mu_s-\rho'\circ\kappa_s}_{\mathsf{TV}}\norm{\ell}_\infty\\
	=&\mathbb{E}_{s\sim\pi_S}\norm{\rho'\circ\lambda\circ\mu_s-\rho'\circ\kappa_s}_{\mathsf{TV}}\norm{\ell}_\infty\\
	\le&\mathbb{E}_{s\sim\pi_S}\norm{\lambda\circ\mu_s-\kappa_s}_{\mathsf{TV}}\norm{\ell}_\infty\\ 
	\le&\mathbb{E}_{s\sim\pi_S}\left[\sqrt{\tfrac{\lnn{2}}{2}D(\kappa_s\|\lambda\circ\mu_s)}\right]\norm{\ell}_\infty\\ 
	\le&\sqrt{\tfrac{\lnn{2}}{2}D(\kappa\|\lambda\circ\mu|\pi_S)}\norm{\ell}_\infty
	\le\sqrt{\tfrac{\lnn{2}}{2}\epsilon}\norm{\ell}_\infty, 
	\end{align*}
	where we have used the data processing inequality for the total variation (TV) distance in the fourth step and Pinsker's inequality in the fifth. The last step follows from the concavity of the square root function. Finally, take a minimum over~$\rho'$ and~$\rho$. This completes the proof.
\end{proof}

\begin{proof} [Proof of Proposition~\ref{lem:positivity_gdefo}]
	Let~$(S,Y,Z)\sim P$. Consider first the case when $UI_o(S;Y\backslash Z)=\delta_o^{\pi}({\mu},{\kappa})$. Then~$UI_o$ is nonnegative by definition. Let~$\lambda^{\ast}\in\mathsf{M}(\Zcal;\Ycal)$ achieve the minimum in~\eqref{eq:weighted_gdefo}. 
	Then~$CI_o$ is nonnegative since
	\begin{align*}
	&I(S;Y|Z)= 
	\sum_s P(s)\sum_z P(z|s)D(P(y|s,z)||P(y|z))\\
	&\ge \sum_s P(s)D\left(\sum_z P(z|s)P(y|s,z)||\sum_z P(z|s)P(y|z)\right)\\
	&=D(P_{Y|S}\|P_{Y|Z}\circ P_{Z|S}|P_S)\\
	&\ge D(P_{Y|S}\|\lambda^{\ast}_{Y|Z}\circ P_{Z|S}|P_S),
	\end{align*}
	where the first inequality follows from the convexity of the Kullback-Leibler divergence and the second inequality follows from the definition of~$\lambda^{\ast}$.
	
	$SI_o$ is nonnegative since
	\begin{align*}
	SI_o(S;Y,Z)&=D(P_{Y|S}\|P_Y|P_S)-D(P_{Y|S}\|\lambda^{\ast}_{Y|Z}\circ P_{Z|S}|P_S)\\
	&\ge D(P_{Y|S}\|P_Y|P_S)-D(P_{Y|S}\|P_{Y}\circ P_{Z|S}|P_S)=0.
	\end{align*} 
	
	The proof for the case when~$UI_o(S;Y\backslash Z)=\delta_o^{\pi}({\kappa},{\mu})+I(S;Y)-I(S;Z)$ or equivalently~$UI_o(S;Z\backslash Y)=\delta_o^{\pi}({\kappa},{\mu})$ by the consistency condition~\eqref{eq:consistency} is similar.
\end{proof}

\begin{proof} [Proof of Lemma~\ref{lem:UIo-zero}]
	If $Y$ has no unique information about~$S$ w.r.t.~$Z$, then $UI(S;Y\backslash Z) = 0$.  Thus, $UI_{o}(S;Y\backslash Z)$
	vanishes by Lemma~\ref{lem:minsyn}. Conversely, assume that $UI_{o}(S;Y\backslash Z)$ vanishes.  By
	Definition~\ref{def:decomp_gdefo}, since $\delta_{o}^{\pi}(\mu,\kappa)$ is a non-negative quantity, it follows that
	$\delta_{o}^{\pi}(\mu,\kappa)=0$.  By definition, $\kappa = \lambda\circ\mu$ for some
	$\lambda\in\mathsf{M}(\Zcal;\Ycal)$, whence $Y$ has no unique information about~$S$ w.r.t.~$Z$.
\end{proof}

\begin{proof} [Proof of Proposition~\ref{lem:positivity_gdefi}]
	The proof is similar to that of Proposition~\ref{lem:positivity_gdefo} and is omitted.
\end{proof}


\begin{proof} [Proof of Proposition~\ref{prop:UIred_equals_gdefi}]
	The proof is direct by noting that $I(S;Y)-I_{S}(Y\searrow Z)=D(\bar{\kappa}\|Q_{Y\searrow Z}|\pi_Y)$ and the fact that  
	\begin{align*} 
	D(\bar{\kappa}\|Q_{Y\searrow Z}|\pi_Y)&=\sum_{y\in\Ycal} \pi_Y(y)\min_{Q\in\conv\left(\{\bar{\mu}_z\}_{z\in\Zcal}\right)} D(\bar{\kappa}_y\|Q)\\
	&=\min_{\bar{\lambda}\in\mathsf{M}(\Ycal;\Zcal)} \sum_{y\in\Ycal} \pi_Y(y) D(\bar{\kappa}_y\|\bar{\mu}\circ\bar{\lambda}_y)
	=\delta_i^{\pi}(\bar{\mu},\bar{\kappa}).\qedhere
	\end{align*} 
\end{proof}

\begin{proof} [Proof of Lemma~\ref{lem:LORR}]
	Let~$(S,S',Y,Y',Z)\sim P'$ and let~$P$ be the~$(S,Y,Z)$-marginal of~$P'$. Let $Q'\in\Delta_{P'}$, and let~$Q$ be the~$(S,Y,Z)$-marginal of~$Q'$. Then~$Q\in\Delta_{P}$. Moreover,
	\begin{align*}
	I_{Q'}(SS';YY'|Z) \ge I_{Q'}(S;Y|Z)
	= I_{Q}(S;Y|Z).
	\end{align*}
	The proof is complete by taking the minimum over~$Q'\in\Delta_{P'}$.
\end{proof}

\begin{proof} [Proof of Lemma~\ref{lem:LO}]
	Consider random variables $S, S', Y, Z$ such that $YZ - S - S'$ is a Markov chain.  Let $P'$ be the marginal
	distribution of $(S',Y,Z)$, and let $P$ be the $(S,Y,Z)$-marginal.  Let
	$Q^{*} = \arg\min_{Q\in\Delta_{P}} I_{Q}(S;Y|Z)$, and let
	\begin{equation*}
	Q^{*\prime}(s',y,z) = \sum_{s}P'(s'|s) Q^{*}(s,y,z).
	\end{equation*}
	Then $Q^{*\prime}\in\Delta_{P'}$.  By definition,
	\begin{align*}
	UI(S;Y\backslash Z) &= I_{Q^{*}}(S;Y|Z)\\
	&\ge I_{Q^{*\prime}}(S';Y|Z)\ge \min_{Q'\in\Delta_{P'}} I_{Q'}(S;Y|Z)
	= UI(S';Y\backslash Z),
	\end{align*}
	where the conditional form of the data processing inequality was used.  This chain of inequalities shows that $UI$
	cannot increase under local operations of~$S$.
	
	Exchanging $Y$ and~$S$ in the above proof shows that the same is true for local operations of~$Y$ (the only slight
	difference occurs when checking that $Q^{*\prime}\in\Delta_{P'}$).
\end{proof}

\begin{proof} [Proof of Lemma~\ref{lem:LOEve}]
	Let $(S,Y,Z)\sim P$ and $(S,Y,Z,Z')\sim P'$.  By
	definition, $P$ is a marginal of~$P'$. Let~$Q\in\Delta_{P}$, and let
	$Q'(s,y,z,z') := Q(s,y,z) P'(z'|s,z)$ if~$P(s,z)>0$ and~$Q'(s,y,z,z')=0$ otherwise. Then~$Q'\in\Delta_{P'}$. Moreover, $Q$ is the $(S,Y,Z)$-marginal of~$Q'$, and~$Y-SZ-Z'$ is a Markov chain w.r.t.~$Q'$. Therefore,
	\begin{align*}
	I_{Q'}(S;Y|ZZ') &= I_{Q'}(SZ';Y|Z) - I_{Q'}(Z';Y|Z)
	\\ &
	\le I_{Q'}(SZ';Y|Z)
	= I_{Q'}(S:Y|Z) + I_{Q'}(Z';Y|S,Z)
	= I_{Q'}(S;Y|Z) = I_{Q}(S;Y|Z).
	\end{align*}
	The proof is complete by taking the minimum over~$Q\in\Delta_{P}$.
\end{proof}

\begin{proof} [Proof of Lemma~\ref{lem:PC}]
	Write~$S'=f(S)$. Let $(S,Y,Z)\sim P$ and $(S,Y,Z,S')\sim P'$. By definition, $P$ is a marginal of~$P'$. Let~$Q\in\Delta_{P}$, and define
	$Q'(s,y,z,s')= Q(s,y,z) P'(s'|s)$ if~$P(s)>0$ and~$Q'(s,y,z,s')=0$ otherwise. Then~$Q'(s,y,s')= Q(s,y) P'(s'|s)=P(s,y) P'(s'|s,y)=P'(s,y) P'(s'|s,y)=P'(s,y,s')$. Similarly, $Q'(s,z,s')=P'(s,z,s')$. Then~$Q'\in\Delta_{P'}$. Moreover, $Q$ is the $(S,Y,Z)$-marginal of~$Q'$,
	Therefore,
	\begin{align*}
	I_{Q'}(S;YS'|ZS') &= I_{Q'}(S;Y|ZS')
	\\ & = I_{Q'}(SS';Y|Z) - I_{Q'}(S';Y|Z)
	\\ &
	\le I_{Q'}(SS';Y|Z)=I_{Q'}(S;Y|Z)=I_{Q}(S;Y|Z)
	\end{align*}
	The proof is complete by taking the minimum over~$Q\in\Delta_{P}$. 
\end{proof}

To prove asymptotic continuity (Theorem~\ref{thm:AC}), we need the following lemma.
\begin{lemma} \label{lem:asymcont-helper}
	Let $P,P'\in\mathbb{P}_{\Scal\times\Ycal\times\Zcal}$.  For any $Q\in\Delta_{P}$ there exists $Q'\in\Delta_{P'}$ with
	$\|Q-Q'\|_1\le 5\|P-P'\|_1$.
\end{lemma}
\begin{proof}
	The signed measure $M = Q + P' - P$ has the same pair margins as $P'$ for the pairs $(S,Y)$ and~$(S,Z)$, and $M$ is
	normalized (that is, $\sum_{s,y,z}M(s,y,z) = 1$). If $M$ is nonnegative, the statement of the lemma is true,
	since $\|M - Q\|_1 = \|P - P'\|_1$ with $Q'=M$. Otherwise there exist $s_{0},y_{0},z_{0}$ with $M(s_{0},y_{0},z_{0}) < 0$. Since
	$\sum_{y}M(s_{0},y,z_{0}) = P'(s_{0},z_{0}) \ge 0$ and $\sum_{z}M(s_{0},y_{0},z) = P'(s_{0},y_{0}) \ge 0$ there exist
	$y_{1}\neq y_{0}$ and $z_{1}\neq z_{0}$ with $M(s_{0},y_{1},z_{0}) > 0$ and $M(s_{0},y_{0},z_{1}) > 0$.  Let
	$\nu = \min\bigl\{M(s_{0},y_{1},z_{0}), M(s_{0},y_{0},z_{1}), |M(s_{0}, y_{0}, z_{0})|\bigr\} > 0$, and consider
	the measure $M'$ defined by
	\begin{align*}
	M'(s_{0}, y_{0}, z_{0}) & = M(s_{0}, y_{0}, z_{0}) + \nu,\\
	M'(s_{0}, y_{1}, z_{0}) & = M(s_{0}, y_{1}, z_{0}) - \nu,\\
	M'(s_{0}, y_{0}, z_{1}) & = M(s_{0}, y_{0}, z_{1}) - \nu,\\
	M'(s_{0}, y_{1}, z_{1}) & = M(s_{0}, y_{1}, z_{1}) + \nu,\\
	M'(s_{0}, y, z) & = M(s, y, z), \qquad\mbox{\rlap{ otherwise.}}
	\end{align*}
	Then $M'$ has the same pair margins as $M$ and $P'$ for the pairs $(S,Y)$ and~$(S,Z)$, and $M$ is normalized.
	Moreover, the absolute sum over the negative entries decreases:
	\begin{align*}
	\sum_{s, y, z: M(s, y, z) < 0} | M(s, y, z) | \ge \sum_{s, y, z: M'(s, y, z) < 0}| M'(s, y, z) | + \nu
	> \sum_{s, y, z: M'(s, y, z) < 0}| M'(s, y, z) |.
	\end{align*}
	Finally, $\|M - M'\|_1\le 4 \nu$.  Iterating the procedure, one obtains a normalized measure $M''$ that has the
	same pair margins as $M$ and $P'$ and that is non-negative.  The triangle inequality shows
	\begin{align*}
	\|M'' - M\|_1 \le 4 \sum_{s, y, z: M(s, y, z) < 0}| M(s, y, z) | \le 4 \| P - P' \|_1.
	\end{align*}
	Thus, $\|M'' - Q\|_1 \le \|M'' - M\|_1 + \|M - Q\|_1 \le 5 \| P - P' \|_1$.  Hence, the statement follows with
	$Q' = M''$.
\end{proof}

\begin{proof} [Proof of Theorem~\ref{thm:AC}]
	Let $P,P'\in\mathbb{P}_{\Scal\times\Ycal\times\Zcal}$, let $Q\in \argmin_{Q\in\Delta_{P}} I_{Q}(S;Y|Z)$.  Choose
	$Q'\in\Delta_{P'}$ as in Lemma~\ref{lem:asymcont-helper}, and let $Q^{*}\in \argmin_{Q\in\Delta_{P'}} I_{Q}(S;Y|Z)$.
	Then
	\begin{align*}
	&{UI}_{P'}(S;Y\backslash Z) - {UI}_{P}(S;Y\backslash Z)\\
	&= I_{Q^{*}}(S;Y|Z) - I_{Q}(S;Y|Z) \\
	&\le I_{Q'}(S;Y|Z) - I_{Q}(S;Y|Z)\\
	&\le 2h'(\epsilon) + \tfrac{5}{2} \epsilon \log\min\{|\Scal|, |\Ycal|\},
	\end{align*}
	where $h'(\epsilon) = \max_{0 \le x\le \min\{\tfrac{5\epsilon}{2},1\}}h(x)$, $h(\cdot)$ is the binary entropy function 
	and where we have used the fact that for any~$P,P'\in\mathbb{P}_{\Scal\times\Ycal\times\Zcal}$, if~$\|P-P'\|_1= 2\epsilon$ then $I_{P'}(S;Y|Z)\le I_P(S;Y|Z)+2h(\epsilon)+\epsilon\log\min\{|\Scal|,|\Ycal|\}$ \cite{RennerW03}. 
\end{proof}

\begin{proof} [Proof of Theorem~\ref{thm:uppbound}]
	The function~$UI$ satisfies additivity (see Lemma~\ref{lem:AD}), asymptotic continuity (see Theorem~\ref{thm:AC}) and the Normalization property. Furthermore,~$UI$ satisfies monotonicity under local operations of Alice and Bob (see Lemma~\ref{lem:LO}) and monotonicity under one-way public communication by Alice (see Lemma~\ref{lem:PC}). Hence, by Theorem~\ref{thm:monotone}, $UI$ is an upper bound to the one-way secret key rate.
\end{proof}

\vspace{.5cm}
\subsection{Optimization problems} \label{app:opt}
The optimization problems in definitions~\ref{def:gdefo}, \ref{def:gdefi} and~\ref{def:decomp_minsyn} of the functions~$ \delta_o^{\pi}$, $\delta_i^{\pi}$, and~$UI$, respectively, are convex programs. Furthermore, the feasible sets in definitions~\ref{def:gdefi} and~\ref{def:decomp_minsyn} have a nice product structure in relation to the corresponding objective functions. 

Given~$(S,Y,Z)\sim P$ and a value~$s\in\Scal$, let~$A_s: \mathbb{P}_{\Ycal\times\Zcal} \to \mathbb{P}_{\Ycal}\times\mathbb{P}_{\Zcal}$ be the linear map that computes the marginal distributions of~$Y$ and~$Z$, given~$P_{YZ|s}$. Each~$Q\in\Delta_P$ in~\eqref{subeq:UIy} has the form~$Q=\pi_S Q_{YZ|S}$ with~$Q_{YZ|S}\in \bigtimes_{s\in\Scal}\Delta_{P,s}$, where
\begin{align}
\Delta_{P,s}:= \big\{Q_{YZ}\in \mathbb{P}_{\Ycal\times\Zcal}\colon Q_{Y}(y)=\kappa_s(y),
\text{ }Q_{Z}(z)=\mu_s(z)\big\},\text{ } s\in\Scal\label{eq:delPs}
\end{align}
is a fiber of~$A_s$ passing through~$P_{YZ|s}$. Then~$\Delta_{P,s}=(P_{YZ|s}+\ker{A_s})\cap\mathbb{P}_{\Ycal\times\Zcal}$. 
As an intersection of an affine space with the probability simplex, $\Delta_{P,s}$ is a polytope. Using a variational representation of the conditional mutual information (a.k.a. the Golden formula), the objective in~\eqref{subeq:UIy} can be written as a double minimization.
\begin{align} \label{eq:CUIOp}
U&I(S;Y\backslash Z)=\min_{Q \in\Delta_P} I_Q(S;Y|Z)\notag\\
=& \min_{Q_{YZ|S} \in\bigtimes_{s\in\Scal}\Delta_{P,s}} \min_{\lambda\in\mathsf{M}(\Zcal;\Ycal)} D(Q_{YZ|S} \|\lambda\times \mu| \pi_S)\notag\\
=& \min_{\lambda\in\mathsf{M}(\Zcal;\Ycal)} \sum_s \pi_S(s) \min_{Q_{YZ|s} \in\Delta_{P,s}} D(Q_{YZ|s}\|\lambda\times \mu_s). 
\end{align}
\cite{CUIfullver} proposed an efficient alternating minimization algorithm to solve~\eqref{eq:CUIOp}. An alternating minimization algorithm recursively fixes one of the two free variables and minimizes the other. When~$\lambda$ is fixed, each summand involves computing an $I$-projection~\cite{csiszarIproj} to the linear family of probability distributions of~$(Y,Z)$ defined by~$\Delta_{P,s}$.
The different summands can be optimized parallely for the different values of~$s\in\Scal$. 

For the weighted input deficiency~\eqref{eq:weighted_gdefi}, 
\begin{align*} 
\delta_i^{\pi}(\bar{\mu},\bar{\kappa})&=\min_{\bar{\lambda}\in\mathsf{M}(\Ycal;\Zcal)} \sum_{y\in\Ycal} \pi_Y(y) D(\bar{\kappa}_y\|\bar{\mu}\circ\bar{\lambda}_y)\\ 
&=\sum_{y\in\Ycal} \pi_Y(y)\min_{Q\in\conv\left(\{\bar{\mu}_z\}_{z\in\Zcal}\right)\subset\mathbb{P}_{\Scal}} D(\bar{\kappa}_y\|Q), 
\end{align*} 
each summand involves computing a $rI$-projection~\cite{csiszarIproj} to a convex set of probability distributions. 
Again, the summands can be optimized separately for the different values of~$y\in\Ycal$.
This is useful in practice in parallelizing the computations.

\vspace{2cm}                      
\bibliographystyle{IEEEtran}
\bibliography{IEEEabrv,allerton2018}

%

\end{document}